\documentclass[runningheads]{llncs}
\usepackage[T1]{fontenc}
\usepackage{amsmath,amssymb}
\usepackage{graphicx}
\usepackage{tikz}
\usepackage[hidelinks]{hyperref}
\usepackage{cleveref}

\newcommand{\CFD}{\mathrm{CFD}}
\newcommand{\CDFD}{\mathrm{CDFD}}
\newcommand{\LGD}{\mathrm{LGD}}
\newcommand{\LFD}{\mathrm{LFD}}
\newcommand{\OPT}{\mathrm{OPT}}
\newcommand{\blocks}{\mathcal{B}}
\newcommand{\ch}{\mathrm{ch}}

\begin{document}

\title{Claw-Free and Line Graph Edge Deletion on Block Graphs}
\titlerunning{Claw-Free and Line Graph Edge Deletion on Block Graphs}
\author{Rajat Adak} 
\authorrunning{Rajat Adak}
\institute{Department of Computer Science and Automation, Indian Institute of Science, Bengaluru, India\\
\email{rajatadak@iisc.ac.in}} 

\maketitle

\begin{abstract}
Edge deletion problems ask for the fewest edges whose deletion gives a graph with a specified property.
We study two such problems on block graphs. First making the graph claw-free, and second making it a line graph.
Since every line graph is claw-free, the second problem is at least as difficult as the first.
We show that on block graphs the two optimum values are always equal, and both are attained by the line graph of a forest.
The main step is a completion argument, in which starting from any claw-free solution, we put back edges to get some nice structure inside each block still maintaining the claw-free property.
This reduces both edge-deletion problems to the same problem of selecting vertices within the blocks of the input graph.
We solve it by dynamic programming in $O(n+m)$ time, which also produces an optimal deletion set and a forest whose line graph is what remains after the completion argument.
We also extend the method to $K_{1,r+1}$-free edge deletion for every fixed $r$.
\end{abstract}

\section{Introduction}\label{sec:intro}

Graph modification problems ask for a minimum set of changes to a graph so that the resulting graph has a desired property or belongs to a prescribed graph class.
 Depending on the allowed operation, one may delete vertices, add edges (\emph{completion}), delete edges, or both add and delete edges (\emph{editing}), see, for example,~\cite{natanzon2001complexity,crespelle2023survey}. In the edge-deletion problems, given a graph $G$ and a target class $\mathcal{C}$, we try to find a smallest set $D\subseteq E(G)$ for which $G-D\in\mathcal{C}$. Even when membership in $\mathcal{C}$ is easy to test, finding the minimum deletion set can be difficult. For instance, $H$-free edge deletion is NP-complete on general graphs whenever the fixed forbidden induced graph $H$ has at least two edges~\cite{aravind2017dichotomy}. This includes claw-free edge deletion. Edge deletion to several tree-like target classes is also hard on general inputs~\cite{koch2024edge}.

For an induced-subgraph condition, deleting edges can create an obstruction as well as remove one. If a vertex has three neighbors with just one edge between them, deleting that edge creates an induced claw. Thus a solution cannot in general be obtained just by destroying the claws present in the original graph. We study how the structure of a \emph{block graph} makes it possible to optimize all deletion problems together.

A \emph{block} is a maximal connected subgraph without a cut vertex. A block on two vertices is called a bridge. A \emph{block graph} is a graph whose blocks are all cliques. It is easy to check that trees and cliques are trivially block graphs. The cliques  of a block graph intersect at cut vertices, and their incidences have a forest structure. This provides a natural setting for dynamic programming, although edges deleted within a clique may leave a subgraph with no obvious structure.

Our two target classes are claw-free graphs and line graphs. A \emph{claw} is an induced $K_{1,3}$, so a graph is \emph{claw-free} precisely when no vertex has three pairwise nonadjacent neighbors. The \emph{line graph} $L(R)$ of a graph $R$ has a vertex for each edge of $R$. Two vertices of $L(R)$ are adjacent when the corresponding edges of $R$ share an endpoint. We call $R$ a \emph{root graph}. it is easy to see that every line graph is claw-free. The neighbors of the vertex corresponding to an edge $xy$ are covered by the clique (in $L(R)$) of edges incident with $x$ and the clique of edges incident with $y$. The containment is strict. For example, the wheel on $6$ vertices,  $W_6 \cong K_1\vee C_5$ is claw-free, but the neighborhood of its hub is a $C_5$ and cannot be covered by two cliques (see \Cref{sec:limits}).

Several related modification problems have been studied on or towards block graphs. The classical characterization of connected claw-free block graphs as line graphs of trees is due to Chartrand, check~\cite[Theorem~8.5]{harary1969graph}. Bonomo et al.~\cite{bonomo2015complexity} gave a polynomial-time algorithm for weighted \emph{cluster deletion} on block graphs, with disjoint unions of cliques as the target. Bonomo-Braberman et al. obtained a linear-time algorithm for claw-free \emph{vertex} deletion on weighted block graphs~\cite{bonomo2024linear}. Dumas et al. studied edge modification \emph{towards} block graphs, with arbitrary input graphs and block graphs as the target. They showed that edge deletion and editing to block graphs are NP-complete, and gave $O(k^2)$-vertex kernels for these variants when parameterized by the number $k$ of changes~\cite{dumas2025polynomial}. For general inputs, Cygan et al. obtained a polynomial kernel for simultaneous claw-and-diamond-free edge deletion, and Eiben and Lochet obtained one for line-graph edge deletion~\cite{cygan2017polynomial,eiben2020polynomial}. Here the input is a block graph and the task is to find a spanning claw-free or line subgraph after deleting least number of edges. We show that every claw-free solution can be completed to a claw-free block graph, and optimize this structure in linear time.

\subsection{Our Result}
For a graph $G$ with $n$ vertices and $m$ edges, let $\CFD(G)$ be the minimum number of edges whose deletion makes $G$ claw-free, that is,
\[
\CFD(G) = \min\left\{|D| \mid D\subseteq E(G), G-D \text{ is claw-free}\right\}.
\]
Let $\LGD(G)$ be the minimum number of edges whose deletion makes $G$ a line graph, that is,
\[
\LGD(G) = \min\left\{
|D| \mid D\subseteq E(G), G-D \text{ is a line graph}\right\}.
\]
We also define $\LFD(G)$ as the minimum number of edges whose deletion makes $G$ the line graph of a forest:
\[
\LFD(G) = \min\left\{ |D| \mid D\subseteq E(G), G-D \cong L(F) \text{ for some forest } F \right\}.
\]
A line graph of a forest is a line graph, and every line graph is claw-free, so for every graph
\begin{equation}\label{ineq}
    \CFD(G)\le\LGD(G)\le\LFD(G).
\end{equation}
The inequalities can be strict on general graphs. On block graphs, all three values coincide.

\begin{theorem}\label{main}
For every block graph $G$,
\[
\CFD(G)=\LGD(G)=\LFD(G).
\]
Moreover, an optimal set $D$ of edges to delete, together with a forest $R$ such that $G-D\cong L(R)$, can be computed in $O(n+m)$ time.
\end{theorem}

For example, let $G$ consist of three triangles $\{c,a_i,b_i\}$, $i=1,2,3$, sharing only $c$. Any single edge deletion leaves a neighbor of $c$ in each triangle, and these three neighbors induce a claw. Deleting $ca_3$ and $cb_3$ leaves two triangles sharing $c$ and a separate edge $a_3b_3$, which is claw-free. Thus the optimum is two and the remaining graph is also the line graph of a forest.

Here is the idea of the proof. Start with any claw-free spanning subgraph $H$ of a block graph $G$. For each block $B$, let $S_B$ consist of the vertices incident with an edge of $H$ inside $B$. If a vertex belonged to three such sets, one of its neighbors from each of the three blocks would form a claw in $H$. Hence each vertex belongs to at most two sets $S_B$.

Complete every $S_B$ to a clique, adding back only edges of the original block $B$. The resulting graph $J$ contains $H$ and remains claw-free, since every vertex has neighbors in at most two completed cliques. In fact, $J$ is the line graph of a forest. To see this represent each completed clique by a vertex of the root graph, and each vertex of $J$ by an edge incident with the one or two clique vertices that contain it, using a private endpoint when needed. The block-vertex forest of $G$ ensures that this root has no cycle.

We are thus left with choosing sets $S_B\subseteq B$ so that no vertex is chosen in more than two blocks while $\sum_B\binom{|S_B|}{2}$ is as large as possible. A dynamic program on the block-vertex forest solves this selection problem in $O(n+m)$ time. The same selection method, with a larger bound on the number of blocks using each vertex, also handles induced $K_{1,r+1}$-free edge deletion for every fixed $r$ (\Cref{sec:stars}) At $r=1$ it specializes to unweighted cluster deletion.

\section{Blocks and the Block-Vertex Incidence graph}\label{sec:blocks}

If we delete a set $D$ of edges from $G$, we get the graph $G-D$ on the same vertex set, called a spanning subgraph of $G$.
Deleting few edges is the same as keeping many, therefore
\[
\CFD(G)=m-\max\{\,|E(H)| : H \text{ is a claw-free spanning subgraph of } G\,\},
\]
and similarly for $\LGD(G)$ and $\LFD(G)$.
From now on we will think about the edges we keep and not the ones deleted, unless otherwise stated.

We can assume $G$ does not contain any isolated vertices. Let $\blocks$ be the set of blocks in $G$. Note that every edge of $G$ lies in exactly one block. Also, two different blocks share at most one vertex, which is then a cut vertex.
The blocks can be found in $O(n+m)$ time by depth-first search.
In a block graph every block $B$ is a clique, so it contains all $\binom{|B|}{2}$ edges between its vertices.

The \emph{block-vertex tree} $T=T(G)$ has the vertex set $V(G)\cup\blocks$, and a vertex $v$ is joined to a block $B$ exactly when $v\in B$.
A cycle in $T$ would yield a cycle in $G$ using edges from distinct blocks, whereas all edges on a cycle lie in one block. Thus $T$ must be a forest. It has one tree per component of $G$. See \Cref{fig:tree} for an example.
If $G$ has $c$ connected components, then $T$ has $n+|\blocks|$ vertices and $n+|\blocks|-c$ edges.
The edges of $T$ are exactly the pairs $(v,B)$ with $v\in B$, hence
\begin{equation}\label{eqn1}
\sum_{B\in\blocks}|B| = n+|\blocks|-c .
\end{equation}
Since every block has at least two vertices, $2|\blocks|\leq\sum_B|B|$ and hence $|\blocks|\le n-c$.
Thus
\begin{equation}\label{eqn2}
\sum_{B\in\blocks}|B| = n+|\blocks|-c \leq 2(n-c).
\end{equation}
So the total size of all blocks is only $O(n)$, even though $m$ can be much larger.

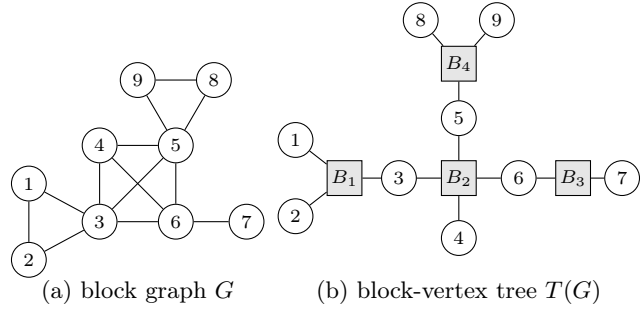
\begin{figure}[t]
\centering
\begin{tikzpicture}[
  scale=0.72,
  v/.style={circle,draw,fill=white,inner sep=0pt,
            minimum size=4.6mm,font=\scriptsize},
  bl/.style={rectangle,draw,fill=gray!20,inner sep=0pt,
             minimum size=4.6mm,font=\scriptsize}
]
\begin{scope}
\node[v] (1) at (-1.3,0.7) {1};
\node[v] (2) at (-1.3,-0.7) {2};
\node[v] (3) at (0,0) {3};
\node[v] (4) at (0,1.4) {4};
\node[v] (5) at (1.4,1.4) {5};
\node[v] (6) at (1.4,0) {6};
\node[v] (7) at (2.7,0) {7};
\node[v] (8) at (2.1,2.6) {8};
\node[v] (9) at (0.7,2.6) {9};

\draw (1)--(2)--(3)--(1);
\draw (3)--(4)--(5)--(6)--(3);
\draw (3)--(5);
\draw (4)--(6);
\draw (6)--(7);
\draw (5)--(8)--(9)--(5);
\node at (0.7,-1.3) {\small (a) block graph $G$};
\end{scope}

\begin{scope}[xshift=6.6cm]
\node[bl] (B2) at (0,0.8) {$B_2$};
\node[v] (t3) at (-1.1,0.8) {3};
\node[bl] (B1) at (-2.1,0.8) {$B_1$};
\node[v] (t1) at (-3.0,1.5) {1};
\node[v] (t2) at (-3.0,0.1) {2};
\node[v] (t6) at (1.1,0.8) {6};
\node[bl] (B3) at (2.1,0.8) {$B_3$};
\node[v] (t7) at (3.0,0.8) {7};
\node[v] (t4) at (0,-0.3) {4};
\node[v] (t5) at (0,1.9) {5};
\node[bl] (B4) at (0,2.9) {$B_4$};
\node[v] (t8) at (-0.7,3.7) {8};
\node[v] (t9) at (0.7,3.7) {9};

\draw (t1)--(B1)--(t2);
\draw (B1)--(t3)--(B2);
\draw (B2)--(t4);
\draw (B2)--(t5)--(B4);
\draw (B4)--(t8);
\draw (B4)--(t9);
\draw (B2)--(t6)--(B3)--(t7);
\node at (0,-1.3) {\small (b) block-vertex tree $T(G)$};
\end{scope}
\end{tikzpicture}
\caption{A block graph and its block-vertex tree. Round vertices are vertices of $G$ and square vertices are the blocks $B_1=\{1,2,3\}$, $B_2=\{3,4,5,6\}$, $B_3=\{6,7\}$ and $B_4=\{5,8,9\}$.}
\label{fig:tree}
\end{figure}

\begin{lemma}\label{lem1}
Let $B_1,\dots,B_t$ be distinct blocks that all contain a vertex $v$, and let $u_i\in B_i\setminus\{v\}$ for each $i$.
Then $u_1,\dots,u_t$ are pairwise nonadjacent.
\end{lemma}
\begin{proof}
Suppose $u_iu_j$ were an edge for some $i\neq j$.
Then the three edges $vu_i$, $vu_j$ and $u_iu_j$ form a triangle, and a triangle lies inside a single block.
This would put $B_i$ and $B_j$ in the same block, contradicting that they are distinct. \qed
\end{proof}

\section{Completing Claw-Free Solutions}\label{sec:completion}

An arbitrary claw-free solution $H$ need not respect the block structure of $G$. The surviving edges in a block need not form a clique. The next lemma shows that we may complete the active vertices inside each block without losing claw-freeness.

\begin{lemma}\label{completion}
Let $G$ be a block graph and let $H$ be any claw-free spanning subgraph of $G$.
For every block $B\in\blocks$ let
\[
S_B=\{v\in B : \text{ there exists an edge of $H$ incident to $v$ that lies in } B\},
\]
and let $J$ be the graph on $V(G)$ obtained from $H$ by adding all the edges inside  $S_B$ for all $B \in \blocks$.
Then $H\subseteq J\subseteq G$, the graph $J$ is claw-free, and every vertex of $G$ belongs to at most two of the sets $S_B$.
\end{lemma}

\begin{proof}
An edge of $H$ lies in a unique block $B$, and both of its endpoints belong to $S_B$, so $H\subseteq J$.
Every block of $G$ is a clique, so every pair inside $S_B\subseteq B$ is an edge of $G$, which gives $J\subseteq G$.

We now show that no vertex $v$ lies in three of the sets, say $S_{B_1}$, $S_{B_2}$ and $S_{B_3}$ for three distinct blocks.
For each $i$, since $v\in S_{B_i}$, there is a neighbour $u_i\in B_i\setminus\{v\}$ with $vu_i\in E(H)$.
By \Cref{lem1} the vertices $u_1,u_2,u_3$ are pairwise nonadjacent in $G$, hence also in $H$.
Together with $v$ they form an induced claw in $H$, a contradiction.
So every vertex belongs to at most two of the sets $S_B$.

Finally, we show that $J$ is claw-free.
The neighbours of a vertex $v$ in $J$ are the vertices of the sets $S_B$ that contain $v$, other than $v$ itself.
There are at most two such sets, and each is a clique in $J$.
So the neighbourhood of $v$ is the union of at most two cliques, and it cannot contain three pairwise nonadjacent vertices. \qed
\end{proof}

\Cref{completion} completes any claw-free solution inside the original blocks without deleting an additional edge. We will refer to adding of the edges to get the above clique structure inside each block as completion. We next construct a forest whose line graph is the completed graph.

\begin{lemma}\label{lem3}
Let $S_B\subseteq B$ be a set for every block $B\in\blocks$ such that every vertex lies in at most two of the sets $S_B$.
Let $J_S$ be the graph on $V(G)$ with edge set $\bigcup_{B}\binom{S_B}{2}$.
Then $J_S$ is the line graph of a forest $R$, and $R$ can be built in $O(n)$ time from the sets $S_B$.
\end{lemma}
\begin{proof}
We can getrid of all sets $S_B$ of size at most one, since they contribute no edges to $J_S$. In the rest of the construction, membership always refers to these retained sets. For each retained $S_B$, create a vertex $r_B$ of $R$.
For every vertex $v$ of $G$ create an edge $e_v$ of $R$ as follows.
If $v$ lies in two sets $S_B$ and $S_C$, let $e_v=r_Br_C$.
If $v$ lies in exactly one set $S_B$, let $e_v$ join $r_B$ to a new private leaf.
If $v$ lies in no set, let $e_v$ be a new edge on two new vertices, disjoint from other vertices of $R$.

No two of these edges are parallel. If $e_u$ and $e_v$ both joined $r_B$ and $r_C$, then $u,v\in B\cap C$, but distinct blocks share at most one vertex. Any cycle in $R$ would use only vertices $r_B$, since all other vertices are private endpoints. A cycle $r_{B_1}r_{B_2}\cdots r_{B_k}$ in $R$ would give an incidence cycle $B_1v_1B_2v_2\cdots B_kv_k$ in $T$, which is a forest. Hence $R$ is a forest.

For every vertex \(v\in V(G)\), the construction creates a unique edge \(e_v\in E(R)\). Hence the map
$$
\phi:V(J_S)\longrightarrow E(R),\qquad \phi(v)=e_v,
$$
is a bijection. It remains to show that $\phi$ preserves adjacency. By definition, $u,v\in V(G)$ are adjacent in $J_S$ if and only if they lie in some retained $S_B$. Every vertex of such an $S_B$ corresponds to an edge incident with $r_B$ and private endpoints are never shared. Thus $u,v\in S_B$ for some retained block if and only if $e_u$ and $e_v$ share an endpoint, or equivalently $e_ue_v\in E(L(R))$. Hence
$$
uv\in E(J_S)
\quad\Longleftrightarrow\quad
e_ue_v\in E(L(R)),
$$
and therefore \(J_S\cong L(R)\).

The construction of $R$ takes $O(n)$ time. Since every vertex belongs to at most two original sets $S_B$, we have
$$\sum_{B\in\mathcal B}|S_B|\le 2n.$$
We go through the sets and discard those of size at most one, create $r_B$ for each retained set, and record only retained memberships for each vertex $v$. This takes $O(n)$ time. Processing each vertex once then creates $e_v$ in constant time from its zero, one, or two retained memberships. Hence $R$ has exactly $n$ edges and $O(n)$ vertices.

\end{proof}

\section{A Common Selection Problem}\label{sec:selection}

We will show that \Cref{completion,lem3} turn all three deletion problems into the same problem. Choose in every block a set of vertices that will will have all possible edges inside it, such that no vertex is chosen in more than two blocks.

\begin{definition}
A \emph{selection} is a family $S=(S_B)_{B\in\blocks}$ of sets $S_B\subseteq B$ such that every vertex of $G$ lies in at most two of the sets $S_B$.
\end{definition}

Given a selection, we keep exactly those edges of $G$ whose two endpoints lie in a common set $S_B$, and we delete all other edges.
Every block is a clique, so all pairs inside $S_B$ are edges of $G$, and different blocks have disjoint edge sets.
Hence the number of kept edges is $\sum_{B}\binom{|S_B|}{2}$.
Let
\begin{equation}\label{eqn4}
\OPT(G)=\max\Big\{\sum_{B\in\blocks}\binom{|S_B|}{2} \mid (S_B)_{B\in\blocks}\text{ is a selection}\Big\}.
\end{equation}

\begin{proposition}\label{same_ineq}
For every block graph $G$,
\[
\CFD(G)=\LGD(G)=\LFD(G)=m-\OPT(G).
\]
\end{proposition}
\begin{proof}
Take an optimal selection $S$.
By \Cref{lem3} the graph $J_S$ is the line graph of a forest, and it is a spanning subgraph of $G$ with $\OPT(G)$ edges.
Deleting the other $m-\OPT(G)$ edges of $G$ therefore gives a line graph of a forest, so $\LFD(G)\le m-\OPT(G)$.

Conversely, let $H$ be any claw-free spanning subgraph of $G$.
By \Cref{completion} the sets $S_B$ form a selection, and $H\subseteq J$, so
\[
|E(H)|\le |E(J)|=\sum_{B}\binom{|S_B|}{2}\le\OPT(G).
\]
So every claw-free solution deletes at least $m-\OPT(G)$ edges, that is, $\CFD(G)\ge m-\OPT(G)$.

Putting these inequalities together with the general inequality in \Cref{ineq} we get
\[
\LFD(G)\le m-\OPT(G)\le\CFD(G)\le\LGD(G)\le\LFD(G),
\]
so all these numbers are equal.
\end{proof}

\begin{corollary}\label{cor:diamond}
Let $\CDFD(G)$ be the minimum number of edges whose deletion makes $G$ both claw-free and diamond-free, where a diamond is $K_4$ minus one edge. For every block graph $G$,
\[
\CDFD(G)=\CFD(G)=\LGD(G)=\LFD(G)=m-\OPT(G).
\]
An optimal deletion set can be found in $O(n+m)$ time.
\end{corollary}
\begin{proof}
Clearly $\CFD(G)\leq\CDFD(G)$. For any two forest edges meeting at $x$, an edge adjacent to both must also meet at $x$, since otherwise those three edges form a cycle. Thus the common neighbors of an adjacent pair in the line graph form a clique, as a result there are no induced diamond. Therefore, the line graph of a forest is claw-free and diamond-free. 

Hence $\CDFD(G)\leq\LFD(G)$, and \Cref{same_ineq} gives the equality. The algorithm in \Cref{sec:dp,sec:linear} returns the corresponding set.
\end{proof}

This proves the equality part of \Cref{main}. We next compute $\OPT(G)$ and an optimal selection in linear time.

\section{Dynamic Programming on the Block-Vertex Tree}\label{sec:dp}

We now show how to compute $\OPT(G)$. Recall that in a selection every vertex can be chosen in at most two blocks. One can think of each vertex as having two \emph{slots}, one for each block in which that vertex may be chosen.

Root every component of the block-vertex forest $T$ at a vertex node. We process the tree from the leaves upward. The main observation is that, when we reach a vertex $v$, we only need to know whether $v$ has already been chosen in its parent block. If it has, then one of its two slots is already used, otherwise both slots are still available. Thus only two states are needed.

For $x\in\{0,1\}$, let $D_v(x)$ denote the maximum number of edges that can be kept in the blocks corresponding the block nodes in the tree below $v$, where $x=1$ means that $v$ is chosen in its parent block and $x=0$ means that it is not. Hence $v$ can be chosen in at most $2-x$ of its child blocks.

Similarly, for a block $B$ with parent vertex $p$, let $M_B(x)$ denote the maximum number of edges that can be kept in $B$ and blocks below it, where $x=1$ means that $p$ is chosen in $B$ and $x=0$ means that it is not.

We first consider a vertex $v$.
The question is how $v$ should use its available slots among its child blocks.
Let $B$ be a child block of $v$. If $v$ is not chosen in $B$, the contribution of the subtree rooted at $B$ is $M_B(0)$. If it is chosen, the contribution is $M_B(1)$. Thus the gain from using one of the slots of $v$ on $B$ is
$$
\Delta_B=M_B(1)-M_B(0).
$$
This gain is nonnegative, since any solution counted by $M_B(0)$ remains feasible when the parent vertex $v$ is also added to the chosen set of $B$.

Therefore, to obtain the best solution at $v$, we start with
$$
\sum_{B\in\ch(v)}M_B(0)
$$
 where $\ch(v)$ denotes the set of child blocks of $v$, and use the available slots on the child blocks with the largest gains. Let $g_1(v)\ge g_2(v)\ge 0$ be the two largest values of $\Delta_B$ among the child blocks of $v$, using $0$ when there are fewer than two child blocks. Then
\begin{equation}\label{eqn5}
D_v(0)=\sum_{B\in\ch(v)}M_B(0)+g_1(v)+g_2(v),
\qquad
D_v(1)=\sum_{B\in\ch(v)}M_B(0)+g_1(v).
\end{equation}

Thus, if $v$ is chosen in its parent block, it loses the opportunity to use its second-best child block. We call this loss the \emph{cost} of choosing $v$ in its parent block, and denote it as
\begin{equation}\label{eqn6}
c_v=D_v(0)-D_v(1)=g_2(v).
\end{equation}

We next consider a block $B$.
Once the values for its child vertices are known, we decide which of these vertices should be chosen in $B$.
Let $p$ be the parent of $B$. The other vertices of $B$, namely $B\setminus\{p\}$, are its children in the rooted block-vertex forest.

Suppose that we choose a set $A\subseteq B\setminus\{p\}$ of these child vertices in $B$. If the parent $p$ is also chosen, then the chosen set in $B$ has $|A|+1$ vertices, otherwise it has $|A|$ vertices. Thus, in state $x\in\{0,1\}$, the number of kept edges contributed by $B$ is $\binom{|A|+x}{2}$.

There is, however, a cost to choosing a child vertex $u$ in $B$. If $u$ is not chosen in $B$, its subtree contributes $D_u(0)$, whereas if it is chosen, its subtree contributes $D_u(1)$. By~\eqref{eqn6}, the loss is exactly $c_u$. Therefore,
\begin{equation}\label{eqn7}
M_B(x)=
\sum_{u\in B\setminus\{p\}}D_u(0) +
\max_{A\subseteq B\setminus\{p\}}\Big\{\binom{|A|+x}{2}-\sum_{u\in A}c_u\Big\}.
\end{equation}

Here we begin with the best value obtainable below every child of $B$ when none of them is chosen in $B$. We then choose some of these children. Doing so creates clique edges inside $B$, but choosing a child $u$ costs $c_u$ because one of its two slots is now used by $B$.

\begin{proposition}\label{dp_prop}
If the values $D_v(x)$ and $M_B(x)$ are computed from the leaves of the rooted block-vertex forest upward using the recurrences above, then for each component with root $r$, the value $D_r(0)$ is the maximum number of edges that can be kept in that component. Consequently,
$$
\OPT(G)=\sum_r D_r(0),
$$
where the sum is over the roots of all components of the block-vertex forest.

Moreover, if we store the choices attaining the maxima in the recurrences, then an optimal selection can be recovered by tracing these choices from the roots downward.
\end{proposition}

\begin{proof}
We argue by induction from the leaves of the rooted block-vertex forest.

Suppose first that the current node is a vertex $v$. Once we decide in which child blocks $v$ is chosen, the subtrees rooted at those child blocks are independent. The only restriction linking them is that $v$ can be chosen in at most $2-x$ child blocks when its state is $x$. Hence the best value is obtained by using the available slots on the child blocks with the largest gains, which gives~\eqref{eqn5}. Thus $D_v(x)$ is the optimum value for the subtree rooted at $v$.

Now suppose the current node is a block $B$. Once we fix the set $A$ of child vertices chosen in $B$, the subtrees rooted at these children are independent. A chosen child $u$ contributes $D_u(1)$ instead of $D_u(0)$, a loss of $c_u$, while the vertices chosen in $B$ contribute $\binom{|A|+x}{2}$ edges inside the clique $B$. Maximizing over all choices of $A$ therefore gives~\eqref{eqn7}, so $M_B(x)$ is the optimum value for the subtree rooted at $B$.

At a root vertex $r$ no parent block is present, so both slots are free and the optimum for that component is $D_r(0)$. Since different components are independent, their values add, giving

$$
\OPT(G)=\sum_r D_r(0).
$$

Finally, by storing the choices attaining each maximum, an optimal selection can be recovered by tracing the decisions from the roots downward.
\end{proof}

\section{Linear-Time Algorithm}\label{sec:linear}

At a vertex node, the dynamic program scans the child blocks, sums their values $M_B(0)$, and keeps the two largest gains $M_B(1)-M_B(0)$. The remaining difficulty is the block computation in \eqref{eqn7}, where it appears to examine every subset of child vertices.

Let $B$ be a block of size $s$, with parent vertex $p$. Its child vertices are $B\setminus\{p\}$. Recall that choosing a child $u$ in $B$ has cost
$c_u=D_u(0)-D_u(1),$
because one of the two available slots of $u$ is then used by $B$.

Suppose first that we decide to choose exactly $j$ child vertices in $B$. In state $x\in\{0,1\}$, the parent $p$ is chosen in $B$ exactly when $x=1$, so the chosen set in $B$ contains $j+x$ vertices. Since $B$ is a clique, these vertices contribute $\binom{j+x}{2}$ kept edges.

For a fixed value of $j$, this contribution depends only on the number of chosen children, and not on which ones are chosen. Thus the best choice of $j$ children is simply to take the $j$ children with the smallest costs. If the costs are ordered as
$$q_1\le q_2\le\cdots\le q_{s-1},$$
then~\eqref{eqn7} can be written as
\begin{equation}\label{eqn8}
M_B(x)=
\sum_{u\in B\setminus\{p\}}D_u(0)
+
\max_{0\le j\le s-1}
\Big\{
\binom{j+x}{2}-\sum_{i=1}^{j}q_i
\Big\}.
\end{equation}

This removes the need to examine all subsets of the children. However, sorting the costs separately at every block would introduce an additional logarithmic factor. We next show that the relevant costs lie in a small range, which allows us to replace sorting by grouping.

\begin{lemma}\label{cutoff}
Let $B$ be a block of size $s$. If a child vertex $v$ of $B$ has cost $c_v\ge s$, then no set attaining the maximum in~\eqref{eqn7} contains $v$.
\end{lemma}

\begin{proof}
Suppose that an optimal set $A$ contains $v$. Remove $v$ from $A$, so it is no longer chosen in the current block $B$. This does not remove $v$ from the graph or its child blocks. Its state changes from $1$ to $0$, so the optimum contribution below $v$ increases from $D_v(1)$ to $D_v(0)$, a gain of $c_v\geq s$.

Inside $B$, removing $v$ from the chosen set loses only the edges joining $v$ to the other chosen vertices of $B$. Since $B$ has $s$ vertices, at most $s-1$ such edges are lost. Hence the total change in the objective is at least
$c_v-(s-1)\ge 1.$

Thus removing $v$ from $A$ strictly improves the solution, contradicting the optimality of $A$.
\end{proof}
Therefore, for a block of size $s$, only costs $0,1,\ldots,s-1$
can be relevant. We create $s$ groups, indexed by these values. Each child $u$ with $c_u<s$ is placed in the group indexed by $c_u$, while children with $c_u\ge s$ are ignored by \Cref{cutoff}. Reading the groups from left to right gives all relevant costs in non-decreasing order without using a comparison sort.

Let
$$q_1\le q_2\le\cdots\le q_h$$
be the relevant costs obtained in this way, and define
$$C_j=\sum_{i=1}^{j}q_i,\qquad C_0=0.$$

Here $C_j$ is exactly the minimum total cost of choosing $j$ child vertices of $B$. Hence
\begin{equation}\label{eqn9}
M_B(x)=
\sum_{u\in B\setminus\{p\}}D_u(0)
+
\max_{0\le j\le h}
\Big\{
\binom{j+x}{2}-C_j
\Big\},
\end{equation}

The groups can be created and scanned in $O(s)$ time, and the values $C_j$ and the two maxima can be computed during the same scan. Thus both states of a block $B$ can be computed in $O(|B|)$ time. We also store, for each state, a value of $j$ attaining the maximum, this will later help us know how many of the cheapest children are chosen in $B$.

 First, the blocks of $G$ and the block-vertex forest $T$ can be constructed in $O(n+m)$ time. We root every component of $T$ at a vertex node and evaluate the dynamic program from the leaves upward.

At a vertex $v$, all its child blocks have already been processed. We scan these child blocks once, compute the sum of their $M_B(0)$ values, and keep track of the two largest gains $M_B(1)-M_B(0)$. This gives $D_v(0)$, $D_v(1)$, and hence the cost $c_v$, in time proportional to the number of child blocks of $v$. At a block $B$, as shown above, both $M_B(0)$ and $M_B(1)$ can be computed in $O(|B|)$ time using the groups.

Thus every block-vertex incidence is examined only a constant number of times. Since
$$\sum_{B\in\mathcal B}|B|=O(n),$$
all vertex and block computations together take $O(n)$ time. By \Cref{dp_prop}, when the bottom-up computation reaches the roots, we obtain
$$\OPT(G)=\sum_r D_r(0).$$

To recover the sets $S_B$, we traverse $T$ once more, now from the roots downward. During the bottom-up computation, we store the choices attaining each maximum.

Every root vertex starts in state $0$, since it has no parent block. Suppose that a vertex $v$ is reached in state $x$. Then $v$ has $2-x$ slots available for its child blocks. We choose $v$ in the child blocks corresponding to the largest gains used in the computation of $D_v(x)$, these blocks receive state $1$, while the remaining child blocks receive state $0$.

Now suppose that a block $B$ with parent $p$ is reached in state $x$. For this state we have stored a value $j$ attaining the maximum in~\eqref{eqn9}. We rebuild the groups and choose $j$ children of smallest cost. These vertices, together with $p$ when $x=1$, form the set $S_B$. Thus
$$S_B=
\begin{cases}
\{\text{the $j$ chosen children}\}, & x=0,\\[1mm]
\{p\}\cup\{\text{the $j$ chosen children}\}, & x=1.
\end{cases}$$

Each chosen child receives state $1$, and every unchosen child receives state $0$. Continuing in this way determines all the sets $S_B$. Since the same block-vertex incidences are scanned only a constant number of times, this top-down reconstruction also takes $O(n)$ time.

During the initial depth-first block decomposition, label each input edge with its unique block. Once the optimal selection is known, mark each vertex-block incidence $v\in S_B$. We scan the edge set once and delete an edge $uv$ labeled $B$ precisely when $u$ and $v$ are not both marked in $S_B$. Each test is constant time, so this produces a deletion set $D$ in $O(n+m)$ time. Since the selection is optimal, the remaining graph has exactly $\OPT(G)$ edges, and hence
$$|D|=m-\OPT(G).$$

By \Cref{same_ineq}, $D$ is optimal for both Claw-Free Edge Deletion and Line Graph Edge Deletion.

Finally, \Cref{lem3} constructs in $O(n)$ time a forest $R$ such that $G-D\cong L(R).$

Therefore an optimal deletion set and the corresponding root forest can be obtained in $O(n+m)$ time. Together with \Cref{same_ineq}, this completes the proof of \Cref{main}.

\section{Why Block Graphs Are Special}\label{sec:limits}

The equality
$$\CFD(G)=\LGD(G)=\LFD(G)$$
relies strongly on the structure of block graphs. There are three features behind the proof. First, every block is a clique, so whenever some vertices of a block remain active in a claw-free solution, we may safely add back all missing edges among them without leaving the original graph. This is the completion step of \Cref{completion}. Second, distinct blocks intersect in at most one vertex. Consequently, if a vertex belongs to several blocks, its neighbors coming from different blocks are pairwise nonadjacent. This is what turns the claw-free condition into the simple restriction that a vertex can belong to at most two selected sets $S_B$. Also, the block-vertex incidence graph is acylic. This acyclic structure is what ensures in \Cref{lem3} that the root graph obtained from the selected sets is itself a forest.

General graphs may not have these special properties. Consider the Wheel graph on $6$ vertices, $W_6$, with $c$ as the hub vertex. The graph is claw-free. Indeed, the neighborhood of $c$ is a $C_5$, which contains no independent set of size three, while for every cycle vertex the centre $c$ is adjacent to both of its other neighbors.

However, $W_6$ is not a line graph. In a line graph, the neighbors of a vertex corresponding to an edge $xy$ are covered by two cliques, root edges incident with $x$ and root edges incident with $y$. In $W_6$, the neighborhood of the centre $c$ is a $C_5$, and two cliques can cover at most four vertices of this cycle. Hence $\CFD(W_6)=0<\LGD(W_6)$.

\section{Extension to Larger Stars}\label{sec:stars}

Now we will see that the slot idea for the vertices works for other stars as well.
Fix an integer $r\ge 1$ and consider deleting the fewest edges so that the graph has no induced $K_{1,r+1}$ (a vertex with $r+1$ pairwise nonadjacent neighbours).
For $r=2$ this is the claw-free problem.

\begin{theorem}\label{thm:stars}
For every fixed $r\ge1$, the minimum number of edges whose deletion makes a block graph $K_{1,r+1}$-free can be computed in $O(n+m)$ time.
\end{theorem}
\begin{proof}
We follow the proof for $r=2$ and list only what changes.
In a $K_{1,r+1}$-free spanning subgraph $H$, a vertex lies in at most $r$ of the sets $S_B$, since $r+1$ of them would give $r+1$ pairwise nonadjacent neighbors by \Cref{lem1}.
The completed graph $J$ is again $K_{1,r+1}$-free, because the neighbourhood of a vertex is the union of at most $r$ cliques.
Conversely, any family of sets in which every vertex lies in at most $r$ of them gives a $K_{1,r+1}$-free graph in the same way.
So the problem is~\eqref{eqn4} with ``at most two'' replaced by ``at most $r$''.
Pad the child-block gains with zeros and let $g_1(v)\ge\cdots\ge g_r(v)\ge0$ be the $r$ largest values. In state $x\in\{0,1\}$, vertex $v$ has $r-x$ free slots, so
\[
D_v(x)=\sum_{B\in\ch(v)}M_B(0)+\sum_{i=1}^{r-x}g_i(v),
\qquad c_v=D_v(0)-D_v(1)=g_r(v).
\]
The block step~\eqref{eqn7} and \Cref{cutoff} are unchanged, so the grouping method still takes $O(|B|)$ time per block, and for fixed $r$ the $r$ largest gains at a vertex can be kept up to date in $O(r)$ time per child.
Recovering the solution works as before.
\end{proof}

For $r=1$, the forbidden graph is $K_{1,2}=P_3$, and $P_3$-free graphs are disjoint unions of cliques. This gives a linear-time algorithm for the unweighted Cluster Deletion problem on block graphs, the weighted version was previously shown polynomial-time solvable by Bonomo, Dur\'an and Valencia-Pabon~\cite{bonomo2015complexity}. For $r=2$ we recover claw-free deletion. Here each vertex belongs to at most two selected sets, permitting the root-forest construction of \Cref{lem3}. For $r>2$, this line-graph construction need not apply.

\section{Conclusion}\label{sec:conclusion}

We gave a common structural and algorithmic treatment of claw-free and line-graph edge deletion on block graphs.
Completing a claw-free solution inside the original blocks costs no extra deleted edges and creates no claw, and the completed graph is the line graph of a forest.
This makes the two optimum values equal, and leaves a selection problem that a two-state dynamic program solves in $O(n+m)$ time.
The same idea covers $K_{1,r+1}$-free deletion for fixed $r$.

Our proof depends on two properties of block graph, two blocks meet in at most one vertex, and the block incidences form a forest.
Therefore, it does not extend directly to chordal graphs.
It would be interesting to find larger graph classes on which the two deletion problems still have the same optimum.

\section*{Declaration of generative AI use}
During the preparation of this manuscript, the authors used ChatGPT to help develop the idea for the linear-time algorithm, particularly the treatment of block nodes in the dynamic program on the block-vertex tree. ChatGPT also assisted with improving the language, TikZ code, and organization of the manuscript. The authors take full responsibility for the mathematical results, proofs, complexity analysis, and final content of the manuscript.

\bibliographystyle{splncs04}
\bibliography{references}

\end{document}